\documentclass{llncs}
\usepackage{graphicx}
\usepackage{amssymb}
\usepackage{latexsym}
\usepackage{subfigure}

 \usepackage{epic}
 \usepackage{eepic}
 \usepackage{epsfig}
 \usepackage{url}
 \usepackage{xspace}
 \usepackage{latexsym}

\newif\ifreport

\newcommand{\nop}[1]{}

\newcommand{\NP}{\textrm{NP}\xspace}

\newcommand{\SigmaP}[1]{\ensuremath{\Sigma_{#1}^P}}

\newcommand{\Or}{\ensuremath{\mathtt{\,v\,}}\xspace}
\newcommand{\derives}{\ensuremath{\mathtt{\ :\!\!-}\ }}

\newcommand{\p}{\ensuremath{{\mathcal{P}}}}
\newcommand{\GP}{\ensuremath{Ground(\p)}}

\newcommand{\BP}{\ensuremath{B_{\p}}}
\newcommand{\UP}{\ensuremath{U_{\p}}}

\newcommand{\EDB}{\ensuremath{E\!D\!B}\xspace}
\newcommand{\IDB}{\ensuremath{I\!D\!B}\xspace}

\newcommand{\R}{\ensuremath{r}}

\newcommand{\HR}{\ensuremath{H(\R)}}
\newcommand{\BR}{\ensuremath{B(\R)}}

\newcommand{\BnR}{\ensuremath{B^-(\R)}}

\newcommand{\naf}{\ensuremath{\mathtt{not}}\xspace}

\newcommand{\dlv}{{\sc DLV}\xspace}

\newcommand{\q}{\ensuremath{{\cal Q}}}

\newcommand{\ground}[1]{\ensuremath{Ground(#1)}}

\newcommand{\head}[1]{\ensuremath{H(#1)}}
\newcommand{\body}[1]{\ensuremath{B(#1)}}
\newcommand{\posbody}[1]{\ensuremath{B^+(#1)}}
\newcommand{\negbody}[1]{\ensuremath{B^-(#1)}}

\newcommand{\atoms}[1]{\ensuremath{Atoms(#1)}}

\newenvironment{simpleprogram}[1][]
   {\vspace{-0.5ex}\begin{itemize}\item[]
      \tt
      \begin{tabbing}
      \code{#1}\ \= \kill
   }
   {\end{tabbing}\end{itemize}\vspace{-3ex}}

\newenvironment{simplealignedprogramstub}[1][]
   {\vspace{-0ex}
      \begin{tabbing}
      #1\kill
   }
   {\end{tabbing}
\vspace{-6ex}}

\newenvironment{sublabeledprogram}[1][]
   {\begin{array}{ll}\setlength{\arraycolsep}{0pt}}
   {\end{array}}

\newcommand{\code}[1]{\ensuremath{#1}}
\newenvironment{dlvcode}
  {\begin{displaymath}\begin{array}{l}}
  {\end{array}\end{displaymath}}

\newcommand{\father}{fath}
\newcommand{\ancestor}{anc}
\newcommand{\brother}{brot}
\newcommand{\related}{rel}

\newcommand{\ASP}{\ensuremath{\rm ASP}}
\newcommand{\ASPSC}{\ensuremath{\rm ASP^{\rm sc}}}

\newcommand{\dat}{\mbox{Datalog}}

\newcommand{\Ans}{{\it Ans}}
\newcommand{\F}{\mathcal{F}}
\renewcommand{\P}{\mathcal{P}}
\newcommand{\Q}{\mathcal{Q}}
\newcommand{\SM}{\mathcal{AS}}

\newcommand{\bcons}{\ensuremath{\models_b}}
\newcommand{\ccons}{\ensuremath{\models_c}}

\newcommand{\qrelation}[3]{\ensuremath{{#1}_{#2}^{#3}}}

\newcommand{\qequiv}[2]{\ensuremath{\qrelation{\equiv}{#1}{#2}}}
\newcommand{\bqequiv}[1]{\ensuremath{\qequiv{#1}{b}}}
\newcommand{\cqequiv}[1]{\ensuremath{\qequiv{#1}{c}}}

\newcommand{\magicRules}{\ensuremath{\mathit{magicRules}}}
\newcommand{\modifiedRules}{$\mathit{modifiedRules}$}

\newcommand{\magic}[1]{\ensuremath{\mathtt{magic}(#1)}}
\newcommand{\dmsqp}{\ensuremath{\DMS(\Q,\p)}}

\newcommand{\killed}[4]{\ensuremath{\mathtt{killed}^{#1}_{#3,#4}(#2)}}
\newcommand{\killedmpmp}{\ensuremath{\killed{M'}{M'}{\Q}{\p}}}

\newcommand{\killedmn}{\ensuremath{\killed{M}{N}{\Q}{\p}}}

\newcommand{\variant}[4]{\ensuremath{\mathtt{var}_{#1,#2}^{#3}(#4)}}

\newcommand{\variantqpi}[1]{\ensuremath{\variant{\Q}{\p}{#1}{I}}}

\newcommand{\magica}{*}

\newcommand{\DMS}{\ensuremath{\mathtt{DMS}}}

\renewcommand{\t}{\bar t}
\newcommand{\s}{\bar s}
\newcommand{\z}{\bar z}

\def\<{\mbox{$\langle$}}
\def\>{\mbox{$\rangle$}}

\newcounter{myenumctr}

\newcommand{\hs}{\hspace{3mm}}

\begin{document}
\pagestyle{plain}

\title{Dynamic Magic Sets for\\ Super-Consistent Answer Set Programs\thanks{This research has been partly supported by Regione Calabria and EU
under POR Calabria FESR 2007-2013 within the PIA
project of DLVSYSTEM s.r.l., and by MIUR under the PRIN project LoDeN. }}

\author{Mario Alviano \and Wolfgang Faber}
\institute{Department of Mathematics, University of Calabria, 87030 Rende (CS), Italy
\\\email{\{alviano,faber\}@mat.unical.it}
}

\maketitle

\begin{abstract}
For many practical applications of \ASP{}, for instance data
integration or planning, query answering is important, and therefore
query optimization techniques for \ASP{} are of great interest. Magic
Sets are one of these techniques, originally defined for Datalog
queries (\ASP\ without disjunction and negation). Dynamic Magic Sets
($\DMS$) are an extension of this technique, which has been proved to
be sound and complete for query answering over \ASP\ programs with
stratified negation.

A distinguishing feature of $\DMS$ is that the optimization can be
exploited also during the nondeterministic phase of \ASP\ engines. In
particular, after some assumptions have been made during the
computation, parts of the program may become irrelevant to a query
under these assumptions. This allows for dynamic pruning of the search
space, which may result in exponential performance gains.

In this paper, the correctness of $\DMS$ is formally established and
proved for brave and cautious reasoning over the class of
super-consistent \ASP\ programs (\ASPSC{} programs). \ASPSC{} programs
guarantee consistency (i.e., have answer sets) when an arbitrary set
of facts is added to them. This result generalizes the applicability
of $\DMS$, since the class of \ASPSC{} programs is richer than
\ASP\ programs with stratified negation, and in particular includes
all odd-cycle-free programs. $\DMS$ has been implemented as an
extension of \dlv, and the effectiveness of $\DMS$ for \ASPSC{}
programs is empirically confirmed by experimental results with this
system.
\end{abstract}

\section{Introduction}\label{sec:introduction}

Answer Set Programming (\ASP) is a powerful formalism for knowledge representation
and common sense reasoning \cite{bara-2002}.
Allowing disjunction in rule heads and nonmonotonic negation in bodies,
\ASP\ can express every query belonging to the complexity class $\rm\SigmaP2$
($\NP^\NP$).
For this reason, it is not surprising that \ASP\ has found several practical
applications, also encouraged by the availability of efficient inference engines, such as \dlv{}
\cite{leon-etal-2002-dlv}, GnT \cite{janh-etal-2000}, Cmodels \cite{lier-2005-lpnmr}, or ClaspD \cite{dres-etal-2008-KR}. As a matter of fact, these systems are
continuously enhanced to support novel optimization strategies, enabling them to be effective over increasingly
larger application domains. Magic Sets are one of these techniques
\cite{ullm-89,banc-etal-86,beer-rama-91}.

The goal of the original Magic Set method (defined in the field of Deductive Databases for Datalog programs,
i.e., disjunction-free positive \ASP\ programs) 
is to exploit the presence of constants in a query
for restricting the possible search space by considering only a subset
of a hypothetic program instantiation, which is sufficient to answer
the query in question. Magic sets are extensions of predicates
that make this restriction explicit. Extending these ideas to
\ASP\ faces a major challenge: While
Datalog programs are deterministic, \ASP\ programs are in general nondeterministic.

There are two basic possibilities how this nondeterminism can be dealt
with in the context of Magic Sets: The first is to consider
\emph{static} magic sets, in the sense that the definition of the
magic sets is still deterministic, and therefore the extension of the
magic set predicates is equal in each answer set.  The second
possibility is to allow \emph{dynamic} magic sets, which also allow
for non-deterministic definitions of magic sets. This means that the
extension of the magic set predicates may differ in various answer
sets, and thus can be viewed as being specialized for different answer
sets. This also mimics the architecture of current \ASP\ systems, which
are divided into a deterministic (grounding) and a non-deterministic
(model search) phase.

In \cite{alvi-etal-2009-TR} the first Dynamic Magic Set ($\DMS$)
method has been proposed and proved correct for \ASP\ with stratified negation.  
In this work, we show that this technique
can be easily extended and shown to be correct for a broader class of
programs, which we call super-consistent \ASP\ programs (\ASPSC{} programs), which includes all stratified and odd-cycle-free programs.
In more detail, the contributions are:

\begin{itemize}

 \item
  We formally establish the correctness of $\DMS$ for \ASPSC{} programs. In particular, we
  prove that the program obtained by the transformation $\DMS$ is
  query-equivalent to the original program. This result holds for both
  brave and cautious reasoning.

 \item
  We have implemented a $\DMS$ optimization module inside the \dlv
  system \cite{leon-etal-2002-dlv}. In this way, we could exploit the
  internal data-structures of the \dlv system and embed $\DMS$ in the
  core of \dlv. As a result, the technique is completely transparent
  to the end user. The implementation is available at \url{http://www.dlvsystem.com/magic/}.

 \item
  We have conducted experiments on a synthetic domain that
  highlight the potential of $\DMS$. We have compared the performance
  of the \dlv system without magic set optimization and
  with $\DMS$. The results show that $\DMS$ can yield drastically better
  performance than the non optimized evaluation.

\end{itemize}

\paragraph{Organization.} In
Section~\ref{sec:preliminaries}, syntax and semantics of \ASP\ 
are introduced. In this section, we also define \ASPSC{} programs.
In Section~\ref{sec:magic}, we show how to apply $\DMS$
to \ASPSC{} programs and formally prove its correctness.  
In Section~\ref{sec:system}, we discuss the
implementation and integration of the Magic Set method within the \dlv
system.  Experimental results are reported in
Section~\ref{sec:experiment}.  Finally,
\nop{ related work is
discussed in Section~\ref{sec:relatedwork}, and} in
Section~\ref{sec:conclusion} we draw our conclusions.

\section{Preliminaries}\label{sec:preliminaries}

In this section, we recall the basics of ASP and introduce the class
of super-consistent ASP programs (\ASPSC{} programs).

\subsection{ASP Syntax and Semantics}

A \emph{term} is either a \emph{variable} or a \emph{constant}. 
If $\tt p$ is a {\em predicate} of arity $k \geq 0$,
and $\tt t_1, \ldots, t_k$ are terms,
then $\tt p(t_1, \ldots, t_k)$ is an {\em atom}\footnote{%
We use the notation $\tt \t$ for a sequence
of terms, for referring to atoms as $\tt p(\t)$.}. 
A {\em literal} is either an atom $\tt p(\t)$ (a positive literal),
or an atom preceded by the {\em negation as failure} symbol $\tt \naf~p(\t)$
(a negative literal).
A {\em rule} $\R$ is of the form
\begin{dlvcode}
\tt p_1(\t_1) \ \Or\ \cdots \ \Or\ p_n(\t_n) \derives 
    q_1(\s_1),\ \ldots,\ q_j(\s_j),\ \naf~q_{j+1}(\s_{j+1}),\ \ldots,\ \naf~q_m(\s_m).
\end{dlvcode}%
where $\tt p_1(\t_1),\ \ldots,\ p_n(\t_n),\ q_1(\s_1),\ \ldots,\ q_m(\s_m)$ 
are atoms and $n\geq 1,$  $m\geq j\geq 0$. The
disjunction $\tt p_1(\t_1) \ \Or\ \cdots \ \Or\ p_n(\t_n)$ is the {\em head} of~\R{}, 
while the conjunction 
$\tt q_1(\s_1),\ \ldots,\ q_j(\s_j),\ \naf~q_{j+1}(\s_{j+1}),\ \ldots,\ \naf~q_m(\s_m)$ 
is the {\em body} of~\R{}.
Moreover, $\HR$ denotes the set of head atoms, while $\BR$ denotes the set of body literals.
We also use $\posbody{\R}$ and $\negbody{\R}$ for denoting
the set of atoms appearing in positive and negative body literals, respectively,
and $\atoms{\R}$ for the set $\HR \cup \posbody{\R} \cup \negbody{\R}$.
A rule $\R$ is normal (or disjunction-free) if $|\HR| = 1$, 
positive (or negation-free) if $\negbody{\R}=\emptyset$,
a {\em fact} if both $\body{\R}=\emptyset$,
$|\HR| = 1$ and no variable appears in $\HR$.

A \emph{program}
$\P$ is a finite set of rules; if all rules in it are positive (resp.\ normal), 
then $\P$ is a positive (resp.\ normal) program.
Odd-cycle-free and stratified programs constitute other two interesting classes of 
programs. A predicate $\tt p$ appearing in the head of a rule $\R$
{\em depends} on each predicate $\tt q$ such that an atom $\tt q(\s)$ belongs to $\BR$; 
if $\tt q(\s)$ belongs to $\posbody{\R}$, $\tt p$ depends on $\tt q$ positively, otherwise negatively. 
A program is \emph{odd-cycle-free} if there is no cycle of dependencies involving an
odd number of negative dependencies, while it is {\em stratified} if each
cycle of dependencies involves only positive dependencies.

Given a predicate $\tt p$, a {\em defining rule} for $\tt p$ is a rule
$\R$ such that some atom $\tt p(\t)$ belongs to $\head{\R}$. If all
defining rules of a predicate $\tt p$ are facts, then $\tt p$ is an
\EDB\ {\em predicate}; otherwise $\tt p$ is an \IDB\ {\em
predicate}\footnote{\EDB\ and \IDB stand for Extensional
Database and Intensional Database, respectively.}. 
Given a program $\p$,  
the set of rules having some IDB predicate in head
is denoted by $\IDB(\p)$, while $\EDB(\p)$ denotes the remaining rules, that is,
$\EDB(\p) = \P \setminus \IDB(\p)$.

The set of constants appearing in a program $\P$ is 
the \emph{universe} of $\P$ and is denoted by $\UP$\footnote{If $\p$ has no constants,
then an arbitrary constant is added to $\UP$.},
while the set of ground atoms constructible from predicates in $\P$ 
with elements of $\UP$ is the \emph{base}
of $\P$, denoted by $\BP$. We call a term (atom, rule, or program) 
\emph{ground} if it does not contain any variable.
A ground atom $\tt p(\t)$ (resp.\ a ground rule $\R_g$) is
an instance of an atom $\tt p(\t')$ (resp.\ of a rule $\R$) if there is a 
substitution $\vartheta$ from the variables in $\tt p(\t')$ (resp.\ in $\R$) 
to $\UP$ such that ${\tt p(\t)} = {\tt p(\t')}\vartheta$ 
(resp.\ $\R_g = \R\vartheta$). 
Given a program $\p$, $\ground{\p}$ denotes the set of all instances
of the rules in $\p$.

An \emph{interpretation} $I$ for a program $\P$ is a subset of $\BP$. A positive ground 
literal $\tt p(\t)$ is true w.r.t.\ an
interpretation $I$ if ${\tt p(\t)}\in I$; otherwise, it is false. 
A negative ground literal $\tt \naf\ p(\t)$ is true w.r.t.\ $I$ 
if and only if $\tt p(\t)$ is false w.r.t.\ $I$. 
The body of a ground rule $\R_g$ is true w.r.t.\ $I$
if and only if all the body literals of $\R_g$ are true w.r.t.\ $I$, that is,
if and only if $\posbody{\R_g} \subseteq I$ and $\negbody{\R_g} \cap I = \emptyset$. 
An interpretation $I$ {\em satisfies} a ground rule $\R_g\in \GP$ if at least one atom
in $\head{\R_g}$ is true w.r.t.\ $I$ whenever the body of $\R_g$ is true w.r.t.\ $I$. An interpretation $I$ is a
\emph{model} of a program $\P$ if $I$ satisfies all the rules in $\GP$. 

Given an interpretation $I$ for a program $\P$, the reduct of $\P$ w.r.t.\ $I$,
denoted $\ground{\p}^{I}$, is obtained by deleting from $\GP$ all 
the rules $\R_g$ with $\negbody{\R_g} \cap I = \emptyset$,
and then by removing all the negative literals from the remaining rules.
The semantics of a program $\P$ is then given by the set $\SM(\P)$ of the answer sets of
$\P$, where an interpretation $M$ is an answer set for $\P$ if and only if
$M$ is a subset-minimal model of $\ground{\P}^M$.

Given a ground atom ${\tt p(\t)}$ and a program $\P$, $\tt p(\t)$ is a cautious
(resp.\ brave)
consequence of $\P$, denoted by $\P \ccons \tt p(\t)$ (resp.\ $\P \bcons \tt p(\t)$), if ${\tt p(\t)} \in M$ for each
(resp.\ some) $M \in \SM(\P)$.
Given a \emph{query}%
\footnote{The queries considered here allow only atoms for simplicity; more complex queries can still be expressed using appropriate rules.
We assume that each constant appearing in $\q$ also appears in $\p$;
if this is not the case, then we can add to $\p$ a fact $\tt p(\t)$
such that $\tt p$ is a predicate not occurring in $\p$
and $\tt \t$ are the arguments of $\q$.}
 $\Q = {\tt g(\t)?}$, $\Ans_c(\Q,\P)$ 
(resp.\ $\Ans_b(\Q,\p)$) denotes the set of all the substitutions $\vartheta$
for the variables of ${\tt g(\t)}$
such that $\P \ccons {\tt g(\t)}\vartheta$ (resp.\ $\P \bcons {\tt g(\t)}\vartheta$). 
Two programs $\P$ and $\P'$ are cautious-equivalent (resp.\ brave-equivalent)
w.r.t.\ a query $\Q$, denoted by $\P\cqequiv{\Q} \P'$ (resp.\ $\P\bqequiv{\Q} \P'$), 
if $\Ans_c(\Q,\P \cup \F) = \Ans_c(\Q,\P' \cup \F)$ 
(resp.\ $\Ans_b(\Q,\P \cup \F) = \Ans_b(\Q,\P' \cup \F)$) is guaranteed for each
set of facts $\F$ defined over the EDB predicates of $\p$ and $\p'$.

\subsection{Super-Consistent \ASP\ Programs}

We now introduce super-consistent ASP programs (\ASPSC{} programs), 
the main class of programs studied in this paper.

\begin{definition}[\ASPSC{} programs]
A program $\p$ is \emph{super-consistent} if, for every set of facts $\F$,
the program $\p \cup \F$ is consistent, that is, $\SM(\p \cup \F) \neq \emptyset$. Let \ASPSC{} denote the set of all super-consistent programs.
\end{definition}

Deciding whether a program $\p$ is \ASPSC{} is computable.
Indeed, if $\p$ is not \ASPSC{}, then there is a set of facts $\F$ such that 
$\p \cup \F$ is inconsistent.
Such an $\F$ can be chosen among all possible sets of ground atoms
constructible by combining predicates of $\p$ with constants in 
$\UP \cup \{\xi_X \mid X \mbox{ is a variable of } \p\}$ 
(assuming different rules have different variable names and $\xi_X$ does not
belong to $\UP$):
If the inconsistency is not due (only) to atoms in $\BP$
but new constant symbols are required, then the choice of these symbols is negligible 
and the possibility to instantiate each variable with a different constant
is sufficient to trigger the inconsistency.

\ASPSC{} programs constitute an interesting class of programs,
properly including odd-cycle-free programs
(hence also stratified programs). 
Indeed, every odd-cycle-free program admits at least one answer set
and remains odd-cycle-free even if an arbitrary set of facts is added
to its rules. On the other hand,
there are programs having odd-cycles that are \ASPSC{}.

\begin{example}\label{ex:aspsc_odd}
Consider the following program:
\begin{dlvcode}
\tt a\ \Or\ b. \qquad \tt a\ \derives\ \naf~a,\ \naf~b.
\end{dlvcode}%
Even if an odd-cycle involving $\tt a$ is present in the dependency graph,
the program above is \ASPSC{}.
Indeed, the first rule assures that the body of the second rule is false in
every model, then annihilating the odd-cycle.
\hfill $\Box$
\end{example}

\section{Magic-Set Techniques}\label{sec:magic}

The Magic Set method is a strategy for simulating the top-down
evaluation of a query by modifying the original program by means of
additional rules, which narrow the computation to what is relevant for
answering the query.  
Dynamic Magic Sets ($\DMS$) are an extension of this technique, 
which has been proved to be sound and complete for query answering over 
\ASP\ programs with stratified negation.

In this section, we first recall the $\DMS$ algorithm, as presented in
\cite{alvi-etal-2009-TR}.
We then show how to apply $\DMS$ to \ASPSC{} programs and formally prove 
the correctness of query answering for this class.

\subsection{Dynamic Magic Sets}
\label{sec:msStrat}

The method of \cite{alvi-etal-2009-TR}\footnote{For
a detailed description of the standard magic set technique 
we refer to \cite{ullm-89}.} is structured in three main phases.

\noindent
\textbf{(1) Adornment.} The key idea is to materialize the binding information for IDB predicates that would be
propagated during a top-down computation,
like for instance the one adopted by Prolog. According to this kind of evaluation, 
all the rules $\R$ such that ${\tt g(\t')} \in \HR$ (where ${\tt g(\t')}\vartheta = \Q$ for some
substitution $\vartheta$) 
are considered in a first step. Then the atoms in $\atoms{\R\vartheta}$ 
different from $\Q$ are considered as new queries and the procedure is iterated.

Note that during this process the information about \emph{bound}
(i.e.\ non-variable) arguments in the query is ``passed'' to the other
atoms in the rule. Moreover, it is assumed that the rule is processed in
a certain sequence, and processing an atom may bind some of its
arguments for subsequently considered atoms, thus ``generating'' and
``passing'' bindings.  Therefore, whenever an atom is processed, each
of its argument is considered to be either \emph{bound} ($\tt b$) or
\emph{free} ($\tt f$).

The specific propagation strategy adopted in a top-down evaluation scheme is called {\em sideways information
passing strategy} (SIPS), which is just a way of formalizing a partial ordering over the atoms of each rule
together with the specification of how the bindings originate and propagate
\cite{beer-rama-91,grec-2003}.
Thus, in this phase, adornments are first created for the query predicate.
Then each adorned predicate is used to propagate its information to the other atoms of the rules defining it
according to a SIPS, thereby simulating a top-down evaluation. 
While adorning rules, novel binding information in the form of yet unseen adorned predicates may be generated, which should be used
for adorning other rules.

\noindent
\textbf{(2) Generation.} The adorned rules are then used to generate
{\em magic rules} defining {\em magic predicates}, which represent the atoms relevant for answering the input query.
The bodies of magic rules contain the atoms required for binding
the arguments of some atom, following the adopted SIPS.

\noindent
\textbf{(3) Modification.} Subsequently, magic atoms are added to the bodies of the adorned rules in order to
 limit the range of the head variables, thus avoiding the inference of facts which are irrelevant for the query. The resulting rules are called {\em modified rules}.

The complete rewritten program consists of the magic and modified rules
(together with the original EDB).  Given a stratified program $\P$, a 
query $\Q$, and the rewritten program $\P'$, $\P$
and $\P'$ are equivalent w.r.t.\ $\Q$, i.e., $\P\bqequiv{\Q} {\P'}$ and $\P\cqequiv{\Q} {\P'}$ hold
\cite{alvi-etal-2009-TR}.

\subsection{Applying $\DMS$ to \ASPSC{} Programs}
\label{sec:ms}

\begin{figure}[t]
 \centering
 \fbox{\hspace{2mm}\parbox{0.88\textwidth}{\scriptsize
  \begin{description}
  \item[Input:] An \ASPSC{} program $\P$, and a query $\Q=\tt g(\t)?$
  \item[Output:] The optimized program $\DMS(\Q,\P)$.
  \item[var]  $S$: \textbf{set} of adorned predicates;\ \modifiedRules$_{\Q,\P}$,\magicRules$_{\Q,\P}$: \textbf{set} of rules;
  \item[begin] \
  \item[]\emph{\phantom{0}1.}\ \ $S$ := $\emptyset$;\ \ \modifiedRules$_{\Q,\P}$ := $\emptyset$;\ \ \magicRules$_{\Q,\P}$ := \{\textbf{\emph{BuildQuerySeed}}($\Q,\p,S$)\};
  \item[]\emph{\phantom{0}2.}\ \ \textbf{while} $S\neq \emptyset$ \textbf{do}
  \item[]\emph{\phantom{0}3.}\ \ \hs $\tt p^\alpha$\ := an element of $S$; $\quad S$ := $S \setminus \{{\tt p^\alpha}\}$;
  \item[]\emph{\phantom{0}4.}\ \ \hs \textbf{for each} rule $\R \in \P$ and \textbf{for each} atom $\tt p(\t)$ $\in \HR$ \textbf{do}
  \item[]\emph{\phantom{0}5.}\, \ \hs \hs $\R^{a}$:=\textbf{\emph{Adorn}}$(\R,{\tt p^\alpha},S)$;
  \item[]\emph{\phantom{0}6.}\, \ \hs \hs \magicRules$_{\Q,\P}$\ := \magicRules$_{\Q,\P}$ \ $\bigcup$ \textbf{\emph{Generate}}$(\R^a)$;
  \item[]\emph{\phantom{0}7.} \ \hs \hs \modifiedRules$_{\Q,\P}$\ := \modifiedRules$_{\Q,\P}$ \ $\bigcup$ $\{$\,\textbf{\emph{Modify}}$(\R^a)$\,$\}$;
  \item[]\emph{\phantom{0}8.}\ \ \hs \textbf{end for}
  \item[]\emph{\phantom{0}9.}\ \ \textbf{end while}
  \item[]\emph{10.}  $\DMS(\Q,\P)$:=\magicRules$_{\Q,\P}$ \ $\cup$ \modifiedRules$_{\Q,\P}$ $\cup\ \EDB(\p)$;
  \item[]\emph{11.}  \textbf{return} $\DMS(\Q,\P)$;
  \item[end.] \
  \end{description}
 }}
 \caption{Dynamic Magic Set algorithm ($\DMS$) for \ASPSC{} programs.}\label{fig:DMS}
\end{figure}

The algorithm $\DMS$ implementing the Magic-Set technique described in the previous
section is reported in Figure~\ref{fig:DMS}.
The algorithm exploits a set $S$ for storing all the adorned predicates to
be used for propagating the binding of the query and, after all the adorned
predicates are processed,
outputs a rewritten program $\DMS(\Q,\P)$ consisting of a set of
\emph{modified} and \emph{magic} rules, stored by means of the sets 
\modifiedRules$_{\Q,\P}$ and \magicRules$_{\Q,\P}$, respectively.

We note that, even if the $\DMS$ method is presented for stratified \ASP\ programs,
this restriction is not required by the algorithm.
Indeed, in \cite{alvi-etal-2009-TR}, stratification is only used to prove
query equivalence of the rewritten program with the original program.
Here we claim that $\DMS$ can be correctly applied to a larger class of programs,
precisely that of \ASPSC{} programs.

We now describe the applicability of $\DMS$ to \ASPSC{} programs, and in the
next section we will prove its correctness for this class of programs.
For illustrating the technique we will use the following running example.

\begin{example}[Related \cite{grec-2003}]\label{running}
A genealogy graph storing information of
relationship (father/brother) among people is given,
from which a non-deterministic ``ancestor'' relation can be derived.
Assuming the genealogy graph is encoded by facts $\tt \related(p_1,p_2)$ 
when $\tt p_1$ is known to be related to $\tt p_2$, that is, 
when $\tt p_1$ is the father or a brother of $\tt p_2$,
the {\em ancestor} relation can be derived by the following \ASPSC{} 
program $\p_{rel}$:
\begin{dlvcode}
\R_1:\quad \tt \father(X,Y)\ \derives\ \related(X,Y),\ \naf~\brother(X,Y). \\
\R_2:\quad \tt \brother(X,Y)\ \derives\ \related(X,Y),\ \naf~\father(X,Y). \\
\R_3:\quad \tt \ancestor(X,Y) \derives\ \father (X,Y). \\
\R_4:\quad \tt \ancestor(X,Y) \derives\ \father(X,Z),\ \ancestor(Z,Y).
\end{dlvcode}%

Given two people $\tt p_1$ and $\tt p_2$, we consider
a query $\q_{rel} = \tt \ancestor(p_1,p_2)?$ asking whether $\tt p_1$ is an ancestor of $\tt p_2$. \hfill $\Box$
\end{example}

The computation starts in step \emph{1} by initializing $S$ and \modifiedRules$_{\Q,\P}$ to the empty set. Then the
function \textbf{\emph{BuildQuerySeed}}$(\Q,\P,S)$ is used for storing the magic seed 
$\tt \magic{\tt g^\alpha(\t)}.$ in \magicRules$_{\Q,\P}$, 
where $\alpha$ is a string having a $\tt b$ in position $i$ if $\tt t_i$ is a constant,
or an $\tt f$ if $\tt t_i$ is a variable.
Intuitively, the magic seed states that atoms matching the input query are relevant.
In addition, \textbf{\emph{BuildQuerySeed}}$(\Q,\P,S)$ adds the adorned predicate
$\tt magic\_g^\alpha$ into the set $S$. 

\begin{example}\label{runningseed}
Given the query $\Q_{rel} = \tt \ancestor(p_1,p_2)?$ and the program $\p_{rel}$, 
\textbf{\emph{BuildQuerySeed}}$(\Q_{rel},\P_{rel},S)$ creates
the fact $\tt magic\_\ancestor^{bb}(p_1,p_2).$ and inserts $\tt \ancestor^{bb}$ in $S$. \hfill $\Box$
\end{example}

The core of the algorithm (steps \emph{2--9}) is repeated until the set $S$ is empty, i.e., until there is no
further adorned predicate to be propagated. In particular, an adorned predicate $\tt p^\alpha$ is removed from
$S$ in step \emph{3}, and its binding is propagated in each (disjunctive) rule $\R \in \P$ of the form
\begin{dlvcode}
\R: \ \tt p(\t) \ \Or\ p_1(\t_1) \ \Or\ \cdots \ \Or\ p_n(\t_n) \derives 
    q_1(\s_1),\ \ldots,\ q_j(\s_j), \\
    \quad\quad\quad\quad\quad\quad\quad\quad\quad\quad\quad\ \
    \tt \naf~q_{j+1}(\s_{j+1}),\ \ldots,\ \naf~q_m(\s_m).
\end{dlvcode}
(with $\tt n\geq 0$) having an atom $\tt p(\t)$ in the head (note that the rule
$\R$ is processed as often as head atoms with predicate
$\tt p$ occur; steps \emph{4--8}).

\noindent \textbf{(1) Adornment.} Step \emph{5} implements the adornment of the rule
according to a fixed SIPS specifically conceived for
disjunctive programs.

\begin{definition}[SIPS]
\label{def:sip2} A {\em SIPS} for a rule $\R$ w.r.t.\ a binding $\tt \alpha$ for an atom ${\tt p(\t)} \in
\HR$ is a pair $(\prec^{\tt p^\alpha(\t)}_r,f^{\tt p^\alpha(\t)}_r)$, where:
\begin{enumerate}
  \item $\prec^{\tt p^\alpha(\t)}_r$ is a strict partial order over the atoms in $\atoms{\R}$, such that:
  \begin{enumerate}
  \item ${\tt p(\t)} \prec^{\tt p^\alpha(\t)}_r {\tt q(\s)}$, for all atoms ${\tt q(\s)}\in \atoms{\R}$
        different from ${\tt p(\t)}$;
  \item for each pair of atoms ${\tt q(\s)} \in (\HR \setminus \{{\tt p(\t)}\}) \cup \BnR$
        and ${\tt b(\z)} \in \atoms{\R}$,
        ${\tt q(\s)} \prec^{\tt p^\alpha(\t)}_r {\tt b(\z)}$ does not hold; and,
  \end{enumerate}

  \item  $f^{\tt p^\alpha(\t)}_r$ is a function assigning to each atom ${\tt q(\s)}\in \atoms{\R}$ a subset of the variables in $\tt \s$---intuitively,
  those made bound when processing ${\tt q(\s)}$.
\end{enumerate}
\end{definition}

The adornments for a rule $\R$ w.r.t.\ an (adorned) head atom $\tt p^\alpha(\t)$ 
are precisely dictated by $(\prec^{\tt p^\alpha(\t)}_r,f^{\tt
p^\alpha(\t)}_r)$; in particular,
a variable $\tt X$ of an atom $\tt q(\s)$ in $\R$ is bound if and only if either:

\begin{enumerate}
\item ${\tt X}\in f^{\tt p^\alpha(\t)}_r({\tt q(\s)})$ with ${\tt q(\s)} = {\tt p(\t)}$; or,

\item ${\tt X}\in f^{\tt p^\alpha(\t)}_r({\tt b(\bar z)})$ for an atom 
${\tt b(\bar z)}\in \posbody{\R}$ such that 
${\tt b(\bar z)} \prec^{\tt p^\alpha(\t)}_r {\tt q(\s)}$ holds.
\end{enumerate}

The function \textbf{\emph{Adorn}}$(\R,{\tt p^\alpha},S)$ produces an adorned disjunctive rule $\R^a$ 
from an adorned predicate $\tt p^\alpha$ and a suitable
unadorned rule $\R$,
by inserting all newly adorned
predicates in $S$. Hence, in step \emph{5} the rule $\R^a$ is of the form
\begin{dlvcode}
\R^a: \ \tt p^\alpha(\t)\,\Or\,p_1^{\alpha_1}(\t_1)\,\Or\,\cdots\,\Or\,p_n^{\alpha_n}(\t_n) \derives 
    q_1^{\beta_1}(\s_1),\ \ldots,\ q_j^{\beta_j}(\s_j), \\
    \quad\quad\quad\quad\quad\quad\quad\quad\quad\quad\quad\quad\quad\ \ 
    \tt \naf~q_{j+1}^{\beta_{j+1}}(\s_{j+1}),\ \ldots,\ \naf~q_m^{\beta_m}(\s_m).
\end{dlvcode}%
where each $\tt \alpha_1, \ldots, \alpha_n, \beta_1, \ldots, \beta_m$ is either
a string representing the bindings defined in 1. and 2. above (for IDB atoms), 
or the empty string (for EDB atoms).

\begin{example}\label{running2}
Let us resume from Example~\ref{runningseed}.
We are supposing the adopted SIPS is passing the bindings
whenever possible, in particular

\[
\begin{array}{lll}
\begin{array}{rcl}
{\tt \father(X,Y)} & \prec^{\tt \father^{bb}(X,Y)}_{\R_1} & {\tt \related(X,Y)}\\
{\tt \father(X,Y)} & \prec^{\tt \father^{bb}(X,Y)}_{\R_1} & {\tt \brother(X,Y)}\vspace{1em}\\
{\tt \father(X,Y)} & \prec^{\tt \father^{bf}(X,Y)}_{\R_1} & {\tt \related(X,Y)}\\
{\tt \father(X,Y)} & \prec^{\tt \father^{bf}(X,Y)}_{\R_1} & {\tt \brother(X,Y)}\\
{\tt \related(X,Y)} & \prec^{\tt \father^{bf}(X,Y)}_{\R_1} & {\tt \brother(X,Y)}\vspace{1em}\\
{\tt \brother(X,Y)} & \prec^{\tt \brother^{bb}(X,Y)}_{\R_2} & {\tt \related(X,Y)}\\
{\tt \brother(X,Y)} & \prec^{\tt \brother^{bb}(X,Y)}_{\R_2} & {\tt \father(X,Y)}\vspace{1em}\\
{\tt \ancestor(X,Y)} & \prec^{\tt \ancestor^{bb}(X,Y)}_{\R_3} & {\tt \father(X,Y)}\vspace{1em}\\
{\tt \ancestor(X,Y)} & \prec^{\tt \ancestor^{bb}(X,Y)}_{\R_4} & {\tt \father(X,Z)}\\
{\tt \ancestor(X,Y)} & \prec^{\tt \ancestor^{bb}(X,Y)}_{\R_4} & {\tt \ancestor(Z,Y)}\\
{\tt \father(X,Z)} & \prec^{\tt \ancestor^{bb}(X,Y)}_{\R_4} & {\tt \ancestor(Z,Y)}\\
\end{array}
& \qquad &
\begin{array}{l}
f^{\tt \father^{bb}(X,Y)}_{\R_1}({\tt \father(X,Y)}) =  \{{\tt X,Y}\}\\
f^{\tt \father^{bb}(X,Y)}_{\R_1}({\tt \related(X,Y)}) =  \{{\tt X,Y}\}\\
f^{\tt \father^{bb}(X,Y)}_{\R_1}({\tt \brother(X,Y)}) =  \{{\tt X,Y}\}\vspace{.75em}\\
f^{\tt \father^{bf}(X,Y)}_{\R_1}({\tt \father(X,Y)}) =  \{{\tt X}\}\\
f^{\tt \father^{bf}(X,Y)}_{\R_1}({\tt \related(X,Y)}) =  \{{\tt X,Y}\}\\
f^{\tt \father^{bf}(X,Y)}_{\R_1}({\tt \brother(X,Y)}) =  \{{\tt X,Y}\}\vspace{.75em}\\
f^{\tt \brother^{bb}(X,Y)}_{\R_2}({\tt \brother(X,Y)}) =  \{{\tt X,Y}\}\\
f^{\tt \brother^{bb}(X,Y)}_{\R_2}({\tt \related(X,Y)}) =  \{{\tt X,Y}\}\\
f^{\tt \brother^{bb}(X,Y)}_{\R_2}({\tt \father(X,Y)}) =  \{{\tt X,Y}\}\vspace{.75em}\\
f^{\tt \ancestor^{bb}(X,Y)}_{\R_3}({\tt \ancestor(X,Y)}) =  \{{\tt X,Y}\}\\
f^{\tt \ancestor^{bb}(X,Y)}_{\R_3}({\tt \father(X,Y)}) =  \{{\tt X,Y}\}\vspace{.75em}\\
f^{\tt \ancestor^{bb}(X,Y)}_{\R_4}({\tt \ancestor(X,Y)}) =  \{{\tt X,Y}\}\\
f^{\tt \ancestor^{bb}(X,Y)}_{\R_4}({\tt \father(X,Z)}) =  \{{\tt X,Z}\}\\
f^{\tt \ancestor^{bb}(X,Y)}_{\R_4}({\tt \ancestor(Z,Y)}) =  \{{\tt Z,Y}\}
\end{array}
\end{array}
\]

When $\tt \ancestor^{bb}$ is removed from the set $S$, $\R_3$ and $\R_4$\footnote{
Note that, according to the SIPS described above, variable $\tt Z$ in $\tt \ancestor(Z,Y)$
is considered bound because of ${\tt \father(X,Z)} \prec^{\tt \ancestor^{bb}(X,Y)}_{\R_4} {\tt \ancestor(Z,Y)}$
and  $f^{\tt \ancestor^{bb}(X,Y)}_{\R_4}({\tt \father(X,Z)}) =  \{{\tt X,Z}\}$. 
Choosing a different SIPS would result in a different (still correct) program.} are adorned:
\begin{dlvcode}
\R_{3}^a:\ \tt \ancestor^{bb}(X,Y)\ \derives\ \father^{bb}(X,Y).\\
\R_{4}^a:\ \tt \ancestor^{bb}(X,Y)\ \derives\ \father^{bf}(X,Z),\ \ancestor^{bb}(Z,Y).
\end{dlvcode}%
The adorned predicates $\tt \father^{bb}$
and $\tt \father^{bf}$ are added to $S$.
Then, $\tt \father^{bb}$ is removed from $S$ and $\R_1$ is adorned:
\begin{dlvcode}
\R_{1,1}^a:\ \tt \father^{bb}(X,Y)\ \derives\ \related(X,Y),\ \naf~\brother^{bb}(X,Y).
\end{dlvcode}%
Thus, $\tt \brother^{bb}$ is added to $S$.
We then remove $\tt \father^{bf}$ from $S$ and adorn $\R_1$:
\begin{dlvcode}
\R_{1,2}^a:\ \tt \father^{bf}(X,Y)\ \derives\ \related(X,Y),\ \naf~\brother^{bb}(X,Y).
\end{dlvcode}%
In this case nothing is added to $S$.
Finally, $\tt \brother^{bb}$ is removed from $S$ and $\R_2$ is adorned:
\begin{dlvcode}
\R_{2}^a:\ \tt \brother^{bb}(X,Y)\ \derives\ \related(X,Y),\ \naf~\father^{bb}(X,Y).
\end{dlvcode}%
\hfill $\Box$
\end{example}

\noindent \textbf{(2) Generation.} The algorithm uses the adorned rules for generating and collecting the
magic rules in step \emph{6}. 
For an adorned atom $\tt p^\alpha(\bar t)$, let $\magic{\tt p^\alpha(\bar t)}$ be its \emph{magic version}
defined as the atom $\tt magic\_p^\alpha(\bar t')$, where $\tt \bar t'$ is obtained from $\tt \bar t$ by
eliminating all arguments corresponding to an $\tt f$ label in $\tt \alpha$, and where $\tt magic\_p^\alpha$ is
a new predicate symbol (for simplicity denoted by attaching the prefix ``$\tt magic\_$'' to the predicate symbol 
$\tt p^\alpha$).
Then, if $\tt q_i^{\beta_i}(\s_i)$ is an adorned atom 
(i.e., $\beta_i$ is not the empty string)
in an adorned rule $\R^a$ having $\tt p^\alpha(\t)$ in head,
\textbf{\emph{Generate}}$(\R^a)$ produces a
magic rule $\R^\magica$ such that (i) $\head{\R^\magica} = \{\magic{\tt q_i^{\beta_i}(\s_i)}\}$ and (ii) $\body{\R^\magica}$
is the union of $\{\magic{\tt p^\alpha(\t)}\}$ and the set of all the atoms 
${\tt q_j^{\beta_j}(\s_j)} \in \atoms{\R}$ such that ${\tt q_j(\s_j)} \prec^{\tt \alpha}_r {\tt q_i(\s_i)}$. 

\begin{example}\label{running3}
In the program of Example~\ref{running2}, the magic rules produced are 
\begin{dlvcode}
\R_{3\phantom{,1}}^\magica:\ \tt magic\_\father^{bb}(X,Y)\ \derives\ magic\_\ancestor^{bb}(X,Y).\\
\R_{4,1}^\magica:\ \tt magic\_\father^{bf}(X)\ \derives\ magic\_\ancestor^{bb}(X,Y).\\
\R_{4,2}^\magica:\ \tt magic\_ \ancestor^{bb}(Z,Y)\ \derives\ magic\_\ancestor^{bb}(X,Y),\ \father(X,Z).\\
\R_{1,1}^\magica:\ \tt magic\_\brother^{bb}(X,Y)\ \derives\ magic\_\father^{bb}(X,Y).\\
\R_{1,2}^\magica:\ \tt magic\_\brother^{bb}(X,Y)\ \derives\ magic\_\father^{bf}(X),\ \related(X,Y).\\
\R_{2\phantom{,1}}^\magica:\ \tt magic\_\father^{bb}(X,Y)\ \derives\ magic\_\brother^{bb}(X,Y).
\end{dlvcode}%
\hfill $\Box$
\end{example}

\noindent \textbf{(3) Modification.} In step \emph{7} the modified rules are generated and collected.
A modified rule $\R'$ is obtained from an adorned rule $\R^a$ by adding to its body
a magic atom $\magic{\tt p^\alpha(\t)}$ for each atom ${\tt p^\alpha(\t)} \in \head{\R^a}$
and by stripping off the adornments of the original atoms.
Hence, the function
\textbf{\emph{Modify}}$(\R^a)$ constructs a rule $\R'$ of the form
\begin{dlvcode}
\R': \ \tt p(\t)\,\Or\,p_1(\t_1)\,\Or\,\cdots\,\Or\,p_n(\t_n) \derives \magic{\tt p^\alpha(\t)},
     \magic{\tt p_1^{\alpha_1}(\t_1)}, \ldots, \\
     \quad\quad\ \ 
     \magic{\tt p_n^{\alpha_n}(\t_n)},
     \tt q_1(\s_1),\ldots,q_j(\s_j),
     \tt \naf~q_{j+1}(\s_{j+1}),\ \ldots,\ \naf~q_m(\s_m).
\end{dlvcode}%

Finally, after all the adorned predicates have been processed, the algorithm outputs the program
$\DMS(\Q,\P)$.

\begin{example}\label{running4}
In our running example, we derive the
following set of modified rules:
\begin{dlvcode}
\R_{3\phantom{,1}}':\ \tt \ancestor(X,Y)\ \derives magic\_\ancestor^{bb}(X,Y),\ \father(X,Y).\\
\R_{4\phantom{,1}}':\ \tt \ancestor(X,Y)\ \derives magic\_\ancestor^{bb}(X,Y),\ \father(X,Z),\ \ancestor(Z,Y).\\
\R_{1,1}':\ \tt \father(X,Y)\ \derives magic\_\father^{bb}(X,Y),\ \related(X,Y),\ \naf~\brother(X,Y).\\
\R_{1,2}':\ \tt \father(X,Y)\ \derives magic\_\father^{bf}(X,Y),\ \related(X,Y),\ \naf~\brother(X,Y).\\
\R_{2\phantom{,1}}':\ \tt \brother(X,Y)\ \derives magic\_\brother^{bb}(X,Y),\ \related(X,Y),\ \naf~\father(X,Y).
\end{dlvcode}%

The optimized program
$\DMS(\Q_{rel},\P_{rel})$ comprises the above modified rules as well as  the magic rules in
Example~\ref{running3}, and the magic seed $\tt magic\_\ancestor^{bb}(p_1,p_2).$ 
(together with the original EDB). \hfill $\Box$
\end{example}

\subsection{Query Equivalence Results}\label{sec:teoria}

We conclude the presentation of the $\DMS$ algorithm by formally proving its 
correctness. This section essentially follows \cite{alvi-etal-2009-TR}, to which
we refer for the details, while here we highlight the necessary considerations 
for generalizing the results of \cite{alvi-etal-2009-TR} to \ASPSC{} queries. 
Throughout this section, we
use the well established notion of unfounded set for disjunctive programs with negation
defined in \cite{leon-etal-97b}. Since we deal with total interpretations, 
represented as the set of atoms interpreted as true, the
definition of unfounded set can be restated as follows.

\begin{definition}[Unfounded sets]
\label{def:unfoundedset} Let $I$ be an interpretation for a program $\p$, and $X \subseteq \BP$
be a set of ground atoms. Then $X$ is an \emph{unfounded set} for $\p$ w.r.t.\ $I$ if and only if for each ground rule
$\R_g \in \ground{\p}$ with $X \cap \head{\R_g} \neq \emptyset$, either $(1.a)$ $\posbody{\R_g} \not\subseteq I$, 
or $(1.b)$ $\negbody{\R_g} \cap I \neq \emptyset$, or
$(2)$ $\posbody{\R_g} \cap X \neq \emptyset$, or $(3)$ $\head{\R_g} \cap (I \setminus X) \neq \emptyset$.
\end{definition}

Intuitively, conditions $(1.a)$, $(1.b)$ and $(3)$ check if the rule is satisfied by $I$ regardless of the atoms in $X$,
while condition $(2)$ assures that the rule can be satisfied by taking the atoms in $X$ as false.
Therefore, the next theorem immediately follows from the characterization of
unfounded sets in~\cite{leon-etal-97b}.

\begin{theorem}\label{theo:unfounded}
Let $I$ be an interpretation for a program $\p$. Then, for any answer 
set $M \supseteq I$ of $\p$, and for each unfounded set $X$ of $\p$ w.r.t.\ $I$, 
$M \cap X = \emptyset$ holds. 
\end{theorem}

We now prove the correctness of the $\DMS$ strategy by showing that it is
\emph{sound} and \emph{complete}.
In both parts of the proof, we exploit the following set of atoms.

\begin{definition}[Killed atoms]
\label{def:killed} Given a model $M$ for $\dmsqp$, and a model $N \subseteq M$ of $\ground{\dmsqp}^{M}$, 
the set $\killedmn$ of the \emph{killed atoms}
w.r.t.\ $M$ and $N$ is defined as: 
$$
\{\,{\tt p(\t)} \in \BP \setminus N \mid \mbox{either } {\tt p}
    \mbox{ is \EDB, or some } \magic{{\tt p^\alpha(\t)}} \mbox{ belongs to } N \,\}.
$$
\end{definition}

Thus, killed atoms are either 
false instances of some EDB predicate, or false atoms which are relevant 
for $\Q$ (since a magic atom exists in $N$).
Therefore, we expect that these atoms are also false in any answer set for $\p$ 
containing $M \cap \BP$. 

\begin{proposition}
\label{prop:killed_unfounded} Let $M$ be a model for $\dmsqp$, 
and $N \subseteq M$ a model of $\ground{\dmsqp}^{M}$. Then $\killedmn$ is an unfounded set
for $\p$ w.r.t.\ $M \cap \BP$.
\end{proposition}
\begin{proof}
See \cite{alvi-etal-2009-TR}, proof of Proposition~3.15.
\hfill $\Box$
\end{proof}

For proving the completeness of the algorithm we provide a construction for passing from an interpretation for $\p$ to one for $\dmsqp$.

\begin{definition}[Magic variant]
\label{def:magic_variant} Let $I$ be an interpretation for $\p$. We define an interpretation $\variantqpi{\infty}$ for
$\dmsqp$, called the magic variant of $I$ w.r.t.\ $\Q$ and $\p$, as the fixpoint of the following sequence:
$$
\begin{array}{lcl}
\variantqpi{0} & = & \EDB(\p) \\
\variantqpi{i+1} & = & \variantqpi{i} \cup \{ {\tt p(\t)} \in I \mid \mbox{some } {\magic{\tt p^\alpha(\t)}} \mbox{ belongs to } \variantqpi{i} \} \\
 & \cup & \{ {\magic{\tt p^\alpha(\t)}} \mid \exists \ \R_g^\magica \in \ground{\dmsqp} \mbox{ such that }\\
 & & \qquad {\magic{\tt p^\alpha(\t)}} \in \head{\R_g^\magica}  
	      \mbox{ and } \posbody{\R_g^\magica} \subseteq \variantqpi{i} \}, \ \ \ \forall i\geq 0
\end{array}
$$
\end{definition}

By definition, for a magic variant $\variantqpi{\infty}$ of an interpretation $I$ for $\p$, $\variantqpi{\infty} \cap \BP \subseteq
I$ holds. More interesting, the magic variant of an answer set for $\p$ is in turn an answer set for $\dmsqp$
preserving the truth/falsity of $\Q\vartheta$, for every substitution $\vartheta$.

\begin{lemma}
\label{lem:magic_variant_minimal_model}
For each answer set $M$ of $\p$, there is an answer set $M'$ of $\dmsqp$ (which is the magic variant of
$M$)  such that, for every substitution $\vartheta$, $\Q\vartheta \in M$ if and only if
$\Q\vartheta \in M'$.
\end{lemma}
\begin{proof}
We can show that $M' = \variantqpi{\infty}$ is an answer set of $\dmsqp$
(see \cite{alvi-etal-2009-TR}, proof of Lemma~3.21).
Thus, since $\Q\vartheta$ belongs either to $M'$ or to $\killedmpmp$, 
for every substitution $\vartheta$,
the claim follows by Proposition~\ref{prop:killed_unfounded}.
\hfill $\Box$
\end{proof}

Proving the soundness of the algorithm requires quite more attention.
Indeed, if the technique is used for a program which is not \ASPSC{}, the 
rewritten program might provide some wrong answer.

\begin{example}\label{ex:unsound}
Consider the program
\begin{dlvcode}
\tt edb(a). \qquad \tt q(X)\ \Or\ p(X)\ \derives\ edb(X). \qquad 
\tt co(X)\ \derives\ q(X),\ \naf~co(X).
\end{dlvcode}%
and the query $\tt q(a)?$.
The program above admits a unique answer set, namely $\{{\tt edb(a),p(a)}\}$.
Applying $\DMS$ will result in the following program:
\begin{dlvcode}
\tt edb(a). \qquad \tt magic\_q^b(a). \qquad \tt magic\_p^b(X)\ \derives\ magic\_q^b(X).\\
\tt magic\_q^b(X)\ \derives\ magic\_p^b(X).\\
\tt q(X)\ \Or\ p(X)\ \derives\ magic\_q^b(X),\ magic\_p^b(X),\ edb(X).
\end{dlvcode}%
The rewritten program has two answer sets, namely
$\{{\tt magic\_q^b(a),magic\_p^b(a),}$ ${\tt edb(a),p(a)}\}$ and 
$\{{\tt magic\_q^b(a),magic\_p^b(a),edb(a),q(a)}\}$.
Therefore, $\tt q(a)$ is a brave consequence of the rewritten program
but not of the original program.
We note that the original program is not \ASPSC{}; indeed,
an inconsistent program can be obtained by adding the fact $\tt q(a)$.
\hfill $\Box$
\end{example}

The soundness of the algorithm for \ASPSC{} programs is proved below.

\begin{lemma}
\label{lem:extending_minimal_models} Let $\Q$ be a query over an \ASPSC{} program $\p$. 
Then, for each answer set $M'$ of $\dmsqp$, there is an answer set $M$ of $\p$ such that, for
every substitution $\vartheta$,
$\Q\vartheta \in M$ if and only if $\Q\vartheta \in M'$.
\end{lemma}
\begin{proof}
Consider the program $\p \cup (M' \cap \BP)$, that is, the program obtained
by adding to $\p$ a fact for each atom in $M' \cap \BP$.
Since $\p$ is \ASPSC{}, there is at least an answer set $M$ for $\p \cup (M' \cap \BP)$.
Clearly $M \supseteq M' \cap \BP$; moreover, we can show that $M$ is an answer set
of $\p$ as well 
(see \cite{alvi-etal-2009-TR}, proof of Lemma~3.16).
Thus, since $\Q\vartheta$ belongs either to $M'$ or to $\killedmpmp$, 
for every substitution $\vartheta$,
the claim follows by Proposition~\ref{prop:killed_unfounded}.
\hfill $\Box$
\end{proof}

From the above lemma, together with Lemma~\ref{lem:magic_variant_minimal_model}, the
correctness of the Magic Set method with respect to query answering directly follows.

\begin{theorem}
\label{theo:dms_equivalence} Let $\p$ be an \ASPSC{} program, 
and let $\Q$ be a query. Then both $\dmsqp \bqequiv{\Q}
\p$ and $\dmsqp \cqequiv{\Q} \p$ hold.
\end{theorem}

\section{Implementation}\label{sec:system}

The Dynamic Magic Set method ($\DMS$) has been implemented and
integrated into the core of \dlv \cite{leon-etal-2002-dlv},
as shown in the architecture reported in Figure~\ref{fig:architecture}.


In our prototype, the $\DMS$ algorithm is applied automatically by default when the user invokes \dlv with {\tt -FB}
(brave reasoning) or {\tt -FC} (cautious reasoning)
together with a (partially) bound query. Magic Sets are not applied by default if the query does not contain any constant.
The user can modify this default behavior by specifying the command-line options {\tt -ODMS} (for applying Magic Sets) or {\tt -ODMS-} (for disabling magic sets).


Within \dlv, \DMS\ is applied immediately after parsing the program
and the query by the {\em Magic Set Rewriter} module.  The rewritten
program is then processed by the {\em Intelligent Grounding} module
and the {\em Model Generator} module using the standard \dlv
implementation. The only other modification with respect to standard
\dlv is for the output and its filtering: For ground queries, the
witnessing answer set is no longer printed by default, but only if
{\tt --print-model} is specified, in which case the magic predicates
are omitted from the output.

\begin{figure}[t]
 \centering
 \includegraphics[width=0.5\textwidth]{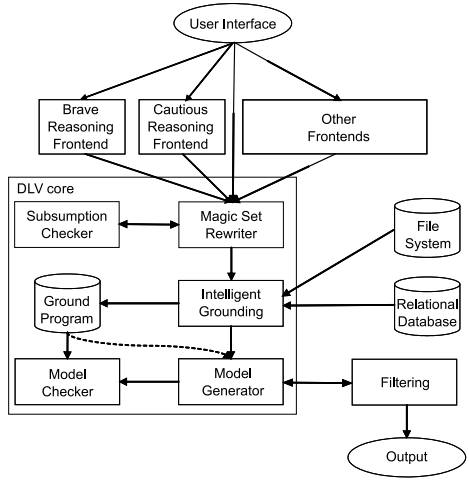}
 \caption{Prototype system architecture.}
 \label{fig:architecture}
\end{figure}

An executable of the \dlv system  supporting the Magic Set optimization is available at \url{http://www.dlvsystem.com/magic/}.

\section{Experimental Results}\label{sec:experiment}

In order to evaluate the impact of the proposed method, we have
compared $\DMS$ with the traditional \dlv evaluation without \emph{Magic Sets}
on several instances of the {\em Related} problem introduced in Section~\ref{sec:magic}.
In our benchmark, the structure of the ``genealogy'' graph consists of
a square matrix of nodes connected as shown in
Figure~\ref{fig:related}, and the instances are generated by varying
the number of nodes (thus the number of persons in the genealogy) of
the graph. We are interested in deciding whether the 
top-leftmost person can be an ancestor of the bottom-rightmost person
(i.e., the benchmark is designed for brave reasoning).
This setting has been used in \cite{grec-2003} for a
disjunctive, negation-free \ASP\ encoding.

\begin{figure}[t]
 \subfigure{\centering \hspace{-2em}\includegraphics[height=135pt]{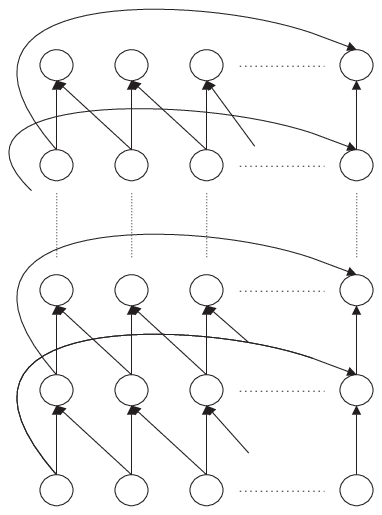}}
 \ \ \hfill
 \subfigure{\centering \includegraphics[height=135pt]{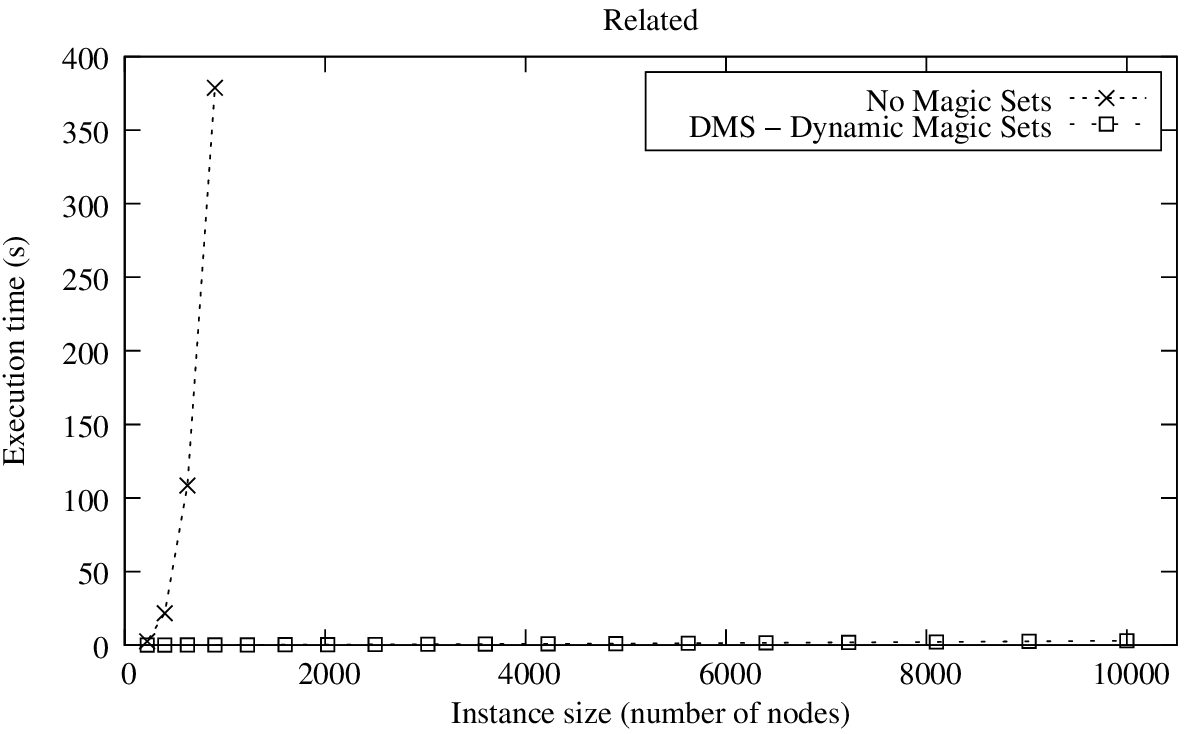}}
 \caption{{\em Related:} Instance structure (left) and average execution time (right).} \label{fig:related}
\end{figure}

The experiments have been performed on a 3GHz Intel$^{\scriptsize\textregistered}$ Xeon$^{\scriptsize\textregistered}$ 
processor system with 4GB RAM under the Debian 4.0 operating system with a GNU/Linux 2.6.23 kernel.
The \dlv prototype used has been compiled using GCC 4.3.3. For each instance, 
we have allowed a maximum running time of 600 seconds (10 minutes) and a 
maximum memory usage of 3GB.

The results for \emph{Related} are reported in
Figure~\ref{fig:related}. Without magic sets, \dlv solves only
the smallest instances, with a very steep increase in execution
time. 
In this case, the exponential computational gain of $\DMS$ over \dlv with no magic sets is due to the 
dynamic optimization of the model search phase resulting from our magic sets
definition. Indeed, $\DMS$ include nondeterministic relevance information
that can be exploited also
during the nondeterministic search phase of \dlv, dynamically
disabling parts of the ground program. In particular, after having
made some choices, parts of the program may no longer be relevant to
the query, but only because of these choices, and the magic atoms
present in the ground program can render these parts satisfied, which
means that they will no longer be considered in this part of the
search. 


\section{Conclusion}\label{sec:conclusion}

The Magic Set method is one of the most well-known techniques for the optimization of positive recursive $\dat$
programs due to its efficiency and its generality.
In this paper, we have elaborated on the issues addressed in \cite{alvi-etal-2009-TR}.
In particular, we have showed the applicability of $\DMS$ for \ASPSC{} programs.
With \DMS,
\ASP\ computations can exploit the information provided by magic set predicates
also during the nondeterministic stable model search, allowing for potentially 
exponential performance gains with respect to unoptimized evaluations.
\nop{Indeed, the definition of our magic set predicates depends on the assumptions
made during the computation, identifying the atoms that are relevant in
the current (partial) scenario.}

We have established the correctness of \DMS\ for \ASPSC{}
by proving that the
transformed program is query-equivalent to the original program.
A strong relationship between magic sets 
and unfounded sets has been highlighted: 
The atoms that are relevant w.r.t.\ a stable model are either true or form an 
unfounded set.

\DMS\ has been implemented in the \dlv system.  Experimental
activities on the implemented prototype system evidenced that our
implementation can outperform the standard evaluation in general also
by an exponential factor.  This is mainly due to the optimization of
the model generation phase, which is specific of our Magic Set
technique. However, we would like to point out that in general we
expect a trade-off between the larger ground program due to the
presence of ground magic atoms and its capability of pruning the
search space.
\nop{ In particular, if the additional ground magic atoms can
cause only an insignificant pruning effect, their presence may also
cause more overhead than the benefit from the pruning. However, in the
experiments that we performed, we did not observe such a behaviour so far.}

As a final point, we would like to point out the relationship of this
work to \cite{fabe-etal-2007-jcss}: There, a Magic Set method for
disjunction-free programs has been defined and proved to be correct
for consistent programs. First, that method will not work for programs
containing disjunction. Second, observe that consistent programs are
not necessarily in \ASPSC{}; indeed the method of
\cite{fabe-etal-2007-jcss} has to take special precautions for relevant parts
of the program that act as constraints (called \emph{dangerous rules})
and thus may impede a relevant interpretation to be an answer set. The
definition of \ASPSC{} implies that programs in this class cannot
contain relevant dangerous rules, 
which allows for the simpler \DMS{} strategy to work correctly.

\bibliographystyle{splncs}

\bibliography{bibtex,magic}

\end{document}